\documentclass[envcountsame,envcountsect,12pt, a4paper]{scrartcl}

\usepackage{amssymb, amsmath}
\usepackage{graphicx}

\usepackage[affil-it]{authblk}

\usepackage{amsthm}
\theoremstyle{plain}
\newtheorem{theorem}{Theorem}[section]
\newtheorem{lemma}[theorem]{Lemma}
\newtheorem{proposition}[theorem]{Proposition}
\newtheorem{corollary}[theorem]{Corollary}

\theoremstyle{definition}

\newtheorem{definition}[theorem]{Definition}

\theoremstyle{remark}
\newtheorem*{remark}{Remark}

\usepackage{mathdots}

\usepackage{pgf}
\usepackage{pgfplots}
\usepackage{tikz}
\usetikzlibrary{arrows,positioning,automata}

\tikzset{node distance=2.8cm, auto, semithick}
\tikzstyle{every state}=[fill=black!15, node distance=3.1cm]

\usepackage{algorithm}
\usepackage{algorithmicx}
\usepackage[noend]{algpseudocode}

\usepackage{complexity}

\def\phi{\varphi}
\def\NAT{\mathbb{N}}
\def\td{\mathit{td}}

\def\edd{\mathop{ed}_d}
\def\level{\mathit{level}}

\title{Graph Isomorphism Parameterized by Elimination Distance to
  Bounded Degree\thanks{Research supported in part by EPSRC grant
    EP/H026835, DAAD grant A/13/05456, and DFG project \textit{Logik,
      Struktur und das Graphenisomorphieproblem}.}}
\author{Jannis Bulian}
\author{Anuj Dawar}
\affil{University of Cambridge Computer Laboratory}

\begin{document}

\maketitle

\begin{abstract}
A commonly studied means of parameterizing graph problems is the
deletion distance from triviality \cite{guo_structural_2004}, which counts
vertices that need to be deleted from a graph to place it in some
class for which efficient algorithms are known.
In the context of graph isomorphism, we define triviality to mean a graph with
maximum degree bounded by a constant, as such graph classes admit
polynomial-time isomorphism tests.  We generalise deletion distance to a measure
we call elimination distance to triviality, based on elimination trees
or tree-depth decompositions.  We establish that graph canonisation,
and thus graph isomorphism, is $\FPT$ when
parameterized by elimination distance to bounded degree, extending
results of Bouland et al.~\cite{bouland_tractable_2012}.
\end{abstract}

\section{Introduction}

The \emph{graph isomorphism problem} ($\GI$) is the problem of
determining, given a pair of graphs $G$ and $H$, whether they are
isomorphic.  This problem has an unusual status in complexity theory
as it is neither known to be in $\P$ nor known to be $\NP$-complete,
one of the few natural problems for which this is the case.
Polynomial-time algorithms are known for a variety of special classes
of graphs.  Many of these lead to natural parameterizations of $\GI$
by means of structural parameters of the graphs which can be used to
study the problem from the point of view of parameterized complexity.  For instance, it is
known that $\GI$ is in $\XP$ parameterized by the genus of the graph,
\cite{Miller:1980iq,Filotti:1980eg}, by maximum degree
\cite{luks_isomorphism_1982,babai_canonical_1983} and by the size of
the smallest excluded minor \cite{Ponomarenko:1991bm}, or more
generally, the smallest excluded topological minor
\cite{GroheMarx2012}.  For each of these parameters, it remains an
open question whether the problem is $\FPT$.  On the other hand, $\GI$
has been shown to be $\FPT$ when parameterized by eigenvalue
multiplicity \cite{Evdokimov:hy}, tree distance width
\cite{Yamazaki:dt}, the maximum size of a simplical component
\cite{Toda:2006bb,Uehara:2005cx} and minimum feedback vertex
set~\cite{Kratsch:ke}.  Bouland et al.~\cite{bouland_tractable_2012}
showed that the problem is $\FPT$ when parameterized by the tree depth
of a graph and extended this result to a parameter they termed
\emph{generalised tree depth}.
 In a recent advance on this, Lokshtanov et
al.~\cite{LokshtanovPPS14} have announced that graph isomorphism is also $\FPT$
parameterized by \emph{tree width}.

Our main result extends the results of Bouland et al.\ and is
incomparable with that of Lokshtanov et al.  We show that graph
canonisation is $\FPT$ parameterized by \emph{elimination distance to
degree $d$}, for any constant $d$.  The structural graph parameter we
introduce is an instance of what Guo et
al.~\cite{guo_structural_2004} call \emph{distance to triviality} and
it may be of interest in the context of other graph problems.

To put this parameter in context, consider the simplest notion of
distance to triviality for a graph $G$: the number $k$ of vertices of
$G$ that must be deleted to obtain a graph with no edges.  This is, of
course, just the size of a minimal vertex cover in $G$ and is a
parameter that has been much studied (see for
instance~\cite{FellowsLMRS08}).  Indeed, it is also quite
straightforward to see that $\GI$ is $\FPT$ when
parameterized by vertex cover number.  Consider two ways this
observation might be strengthened.  The first is to relax the notion
of what we consider to be ``trivial''.  For instance, as there is, for
each $d$, a polynomial time algorithm deciding $\GI$ among
graphs with maximum degree $d$, we may take this as our trivial base
case.  We then parameterize $G$ by the number $k$ of vertices
that must be deleted to obtain a subgraph of $G$ with maximum degree
$d$.  This yields the parameter \emph{deletion distance to bounded
degree}, which we consider in Section~\ref{S:deletion_distance} below.
Alternatively, we relax the notion of ``distance'' so that rather than
considering the sequential deletion of $k$ vertices, we consider the
recursive deletion of vertices in a tree-like fashion.  To be precise,
say that a graph $G$ has \emph{elimination distance} $k+1$ from
triviality if, in each connected component of $G$ we can delete a
vertex so that the resulting graph has distance $k$ to triviality.  If
triviality is understood to mean the empty graph, this just yields a
definition of the tree depth of $G$.  In our main result, we combine
these two approaches by parameterizing $G$ by the elimination distance
to triviality, where a graph is trivial if it has maximum degree $d$.
We show that, for any fixed $d$, this gives a structural parameter on
graphs for which graph canonisation is $\FPT$.  Along the way, we
establish a number of characterisations of the parameter that may be
interesting in themselves.  The key idea in the proof is the
separation, in a canonical way. of any graph of elimination distance $k$ to degree $d$ into
two subgraphs, one of which has degree bounded by $d$ and the other
tree-depth bounded by a function of $k$ and $d$.  It should be noted
that the parameter termed \emph{generalised tree depth}
in~\cite{bouland_tractable_2012} can be seen as a special case of
elimination distance to degree 2.

A central technique used in the proof is to construct, from a graph
$G$, a term (or equivalently a labelled, ordered tree) $T_G$ that is
an isomorphism invariant of the graph $G$.  It should be noted that
this general method is widely deployed in practical isomorphism tests
such as McKay's graph isomorphism testing program ``nauty''
\cite{McKay:1981ug,McKay:2014uw}.  The recent advance by Lokshtanov et
al.~\cite{LokshtanovPPS14} is also based on such an approach.

In Section~\ref{S:prelim} we recall some definitions from graph theory
and parameterized complexity theory.
Section~\ref{S:deletion_distance} introduces the notion of deletion
distance to bounded degree and presents a kernelisation procedure that
allows us to decide isomorphism.  In
Section~\ref{S:elimination_distance} we introduce the main parameter
of our paper, elimination distance to bounded degree, and establish
its key properties.  The main result on $\FPT$ graph canonisation is
established in Section~\ref{S:elimination_distance_alg}.

\section{Preliminaries}
\label{S:prelim}

Parameterized complexity theory is a two-dimensional approach to the study of the complexity of computational problems. A \emph{language} (or
\emph{problem}) $L$ is a set of strings $L \subseteq \Sigma^*$ over a
finite alphabet $\Sigma$. A \emph{parameterization} is a function $\kappa :
\Sigma^* \to \mathbb{N}$. We say that $L$ is \emph{fixed-parameter
  tractable} with respect to $\kappa$ if we can decide whether an
input $x \in \Sigma^*$ is in $L$ in time $O(f(\kappa(x)) \cdot |x|^c)$, where $c$
is a constant and $f$ is some computable function.
For a thorough discussion of the subject we refer to the books by
Downey and Fellows~\cite{Downey:2012vk}, Flum and
Grohe~\cite{Flum:2006vj} and Niedermeier~\cite{Niedermeier:2006ei}.

A \emph{graph} $G$ is a set of vertices $V(G)$ and a set of edges $E(G)
\subseteq V(G) \times V(G)$.
We will usually assume that graphs are loop-free and undirected,
i.e. that $E$ is irreflexive and symmetric.
If $E$ is not symmetric, we call $G$ a \emph{directed graph}. We
mostly follow the notation in Diestel~\cite{Diestel:2000vm}.

If $v \in G$ and $S \subseteq V(G)$, we write $E_G(v, S)$ for the set of
edges $\{vw \mid w \in S\}$ between $v$ and $S$.

The \emph{neighbourhood} of a vertex $v$ is $N_G(v) := \{w \in V(G) \mid vw
\in E(G)\}$. The \emph{degree} of a vertex $v$ is the size of its
neighbourhood $\deg_G(v) := |N_G(v)|$.
For a set of vertices $S \subseteq V(G)$ its neighbourhood is defined
to be $N_G(S) := \bigcup_{v \in S} N_G(v)$.
The \emph{degree} of a graph $G$ is the maximum degree of its
vertices $\Delta(G) := \max\{\deg_G(v) \mid v \in V(G)\}$. If it is
clear from the context what the graph is, we will sometimes omit the subscript.

A subgraph $H$ of $G$ is a graph with vertices $V(H) \subseteq V(G)$
and edges $E(H) \subseteq (V(H) \times V(H)) \cap E(G)$.
If $A \subseteq V(G)$ is a set of vertices of $G$, we write $G[A]$ for
the subgraph \emph{induced} by $A$, i.e. $V(G[A]) = A$ and $E(G[A]) = E(G)
\cap (A \times A)$. If $A$ is a subset of $V(G)$, we write $G
\setminus A$ for $G[V(G) \setminus A]$. For a vertex $v \in V(G)$,
we write $G \setminus v$ for $G \setminus \{v\}$.

A vertex $v$ is said to be \emph{reachable} from a vertex $w$ in $G$
if $v=w$ or if there is a sequence of edges $a_1a_2, \dots,
a_{s-1}a_s\in E(V)$ with the $a_i$ pairwise distinct and $w=a_1$ and $v=a_s$.  We call the subgraph $P$ of $G$ with
vertices $V(P) = \{a_1, \dots, a_s\}$ and edges
$E(P) = \{a_1a_2, \dots, a_{s-1}a_s\}$ a \emph{path from
$w$ to $v$}.

Let $H$ be a subgraph of $G$ and $v, w \in V(G)$. A \emph{path through
  $H$ from $w$ to $v$} is a path $P$ from $w$ to $v$ in 
$G$ with all vertices, except possibly the endpoints, in $V(H)$,
i.e. $(V(P) \setminus \{v, w\}) \subseteq V(H)$.

It is easy to see that for undirected graphs reachability
defines an equivalence relation on the vertices of $G$. A subgraph
of an undirected graph induced by a reachability class is called a
\emph{component}.

Two graphs $G$, $G'$ are \emph{isomorphic} if there is a bijection
$\phi : V(G) \to V(G')$ such that for all $v, w \in V(G)$ we have that
$vw \in E(G)$ if and only if $\phi(v)\phi(w) \in E(G')$. We write $G
\cong G'$ if $G$ and $G'$ are isomorphic.  We write $\GI$ to denote
the problem of deciding, given $G$ and $G'$ whether $G\cong G'$.

A \emph{(k-)colouring} of a graph $G$ is a map $c : V(G) \to \{1, \dots,
k\}$ for some $k \in \NAT$. We call a graph together with a colouring
a \emph{coloured} graph.
Two coloured graphs $G, G'$ with respective colourings $c : V(G) \to
\{1, \dots, k\}, c' : V(G') \to \{1, \dots, k\}$ are \emph{isomorphic} if there is a bijection
$\phi : V(G) \to V(G')$

such that:
\begin{itemize}
\item for all $v, w \in V(G)$ we have that $vw \in E(G)$ if and only
  if $\phi(v)\phi(w) \in E(G')$;
\item for all $v \in V(G)$, we have that $c(v) = c'(\phi(v))$.
\end{itemize}

Note that we require the colour classes to match exactly, and do not
allow a permutation of the colour classes.

Let $\C$ be a class of (coloured) graphs closed under isomorphism. A
\emph{canonical form for $\C$} is a function $F : \C \to \C$ such that
\begin{itemize}
\item for all $G \in \C$, we have that $F(G) \cong G$;
\item for all $G, H \in \C$, we have that $G \cong H$ if, and only if,
  $F(G) = F(H)$.
\end{itemize}

Recall that a \emph{partial order} is a binary relation $\leq$ on a set $S$ which is
reflexive, antisymmetric and transitive. If $\leq$ is a partial order
on $S$, and for each element $a \in S$, the set $\{b \in S \mid b \leq
a\}$ is totally ordered by $\leq$, we say $\leq$ is a \emph{tree
  order}. (Note that the covering relation of a tree order is not
necessarily a tree, but may be a forest.)

\begin{definition}\label{def:elimination-order}
An \emph{elimination order} $\leq$ is a tree order on the vertices of
a graph $G$, such that for each edge $uv \in E(G)$ we have either $u
\leq v$ or $v \leq u$.
\end{definition}
We say that an order has \emph{height} $k$ if the length of the
longest chain in it is $k$.

We write $\td{G}$ for the \emph{tree-depth} of $G$, 
which is defined as follows
\[
\td(G) := \begin{cases}
 0,
 & \text{if }V(G) = \emptyset; \\
 1 + \min\{\td(G \setminus v) \mid v \in V(G)\},
 & \text{if $G$ is connected;} \\
 \max\{\td(H) \mid H \text{ a component of $G$}\},
 & \text{otherwise.}
\end{cases}
\]
Note that there is an elimination order $\leq$ of height $k$ for a
graph $G$ if, and only if, $\td(G) \leq k$.
\section{Isomorphism on bounded-degree graphs}

In this section we collect some well known results about isomorphism
tests and canonisation of bounded degree graphs that we will use.
Luks~\cite{luks_isomorphism_1982} shows that
isomorphism of bounded-degree graphs is decidable in polynomial time.
This result extends, by an easy reduction, to \emph{coloured} graphs
of bounded-degree.  For completeness, we present this reduction
explicitly. 

\begin{proposition}
The isomorphism problem for coloured graphs can be reduced to $\GI$ in
polynomial time. 
\end{proposition}
\begin{proof}
Let $G, G'$ be graphs and let $c, c' : V(G) \to \{1, \dots, k\}$ be
colourings of $G, G'$ respectively for some $k \in \NAT$.

We define $H$ to be the graph whose vertices include $V(G)$ and,
additionally, for each $v \in V(G)$, $c(v)+1$ new vertices
$u^v_1,\ldots,u^v_{c(v)+1}$.  The edges of $H$ are the edges $E(G)$
plus additional edges so that the vertices $v$ and
$u^v_1,\ldots,u^v_{c(v)+1}$ form a simple cycle of length $c(v)+2$.
We obtain $H'$ in a similar way from $G'$.

We claim that $G \cong G'$ if, and only if, $H \cong H'$.
Clearly, if $G \cong G'$ and $\phi$ is an isomorphism witnessing this,
it can be extended to an isomorphism from $H$ to $H'$ by mapping
$u^v_i$ to $u^{\phi(v)}_i$.
For the converse, suppose $H \cong H'$ and let $\phi : H \to H'$ be
an isomorphism.  We use it to define an isomorphism $\phi'$ from $G$
to $G'$.   Note that, if $v \in V(G)$ is not an isolated vertex
of $G$, then it has degree at least 3 in $H$.  Since $\phi(v)$ has the
same degree, it is in $V(G')$, and we let $\phi'(v) = \phi(v)$.  If
$v$ is an isolated vertex of $G$, then its component in $H$ is a
simple cycle of length $c(v)+2$.  The image of this component under
$\phi$ is a simple cycle of $H'$ which must contain exactly one vertex
$v'$ of $V(G')$.  We let $\phi'(v)=v'$.  It is easy to see that there
is an edge between $v_1,v_2$ in $G$ if, and only if, there is an edge
between $\phi'(v_1)$ and $\phi'(v_2)$ in $G'$.  To see that $\phi'$
also preserves colours, note that $\phi$ must map the cycle containing
$u^v_{c(v)}$ to the cycle containing $u^{\phi'(v)}_{c(\phi'(v))}$ and
therefore $c(v) = c(\phi'(v))$. 
\end{proof}
\begin{remark}
Note that the construction in the proof increases the degree of each
vertex by $2$, so if $G$ and $G'$ are graphs of degree $d$, then $H,
H'$ are graphs of degree $d+2$.
\end{remark}

As Luks~\cite{luks_isomorphism_1982} proves that isomorphism of bounded
degree graphs can be decided in polynomial time, we have the
following:
\begin{theorem} \label{T:bdd_iso}
We can test in polynomial time whether two (coloured) graphs with
maximal degree bounded by a constant are isomorphic.
\end{theorem}

Babai and Luks~\cite{babai_canonical_1983} give a
polynomial time canonisation algorithm for bounded degree graphs.
Just as above we can reduce canonisation of coloured
bounded degree graphs to the bounded degree graph canonisation
problem.

\begin{theorem}\label{T:bdd_canon}
Let $\C$ be a class of (coloured) bounded degree graphs closed under
isomorphism. Then there is a canonical form $F$ for $\C$ that allows
us to compute $F(G)$ in polynomial time.
\end{theorem}

\section{Deletion distance to bounded degree}
\label{S:deletion_distance}

We first study the notion of deletion distance to bounded degree and
establish in this section that graph isomorphism is FPT with this
parameter.  Though the result in this section is subsumed by the more
general one in Section~\ref{S:elimination_distance_alg}, it provides a
useful warm-up and a tighter, polynomial kernel.  In the present
warm-up we only give an algorithm for the graph isomorphism problem,
though the result easily holds for canonisation as well (and this
follows from the more general result in
Section~\ref{S:elimination_distance_alg}).  The notion of deletion
distance to bounded degree is a particular instance of the general
notion of distance to triviality introduced by Guo et
al.~\cite{guo_structural_2004}.  In the context of graph isomorphism, we have chosen triviality to mean graphs of bounded
degree.

\begin{definition}
A graph $G$  has \emph{deletion distance $k$ to degree $d$} if
there are $k$ vertices $v_1, \dots, v_k \in V(G)$ such that $G
\setminus \{v_1, \dots, v_k\}$ has degree $d$. We call the
set $\{v_1, \dots, v_k\}$ a \emph{$d$-deletion set}.
\end{definition}
\begin{remark}
 To say that $G$ has deletion distance $0$ from degree
$d$ is just to say that $G$ has maximum degree $d$.  Also note that
if $d=0$, then the $d$-deletion set is just a vertex cover and the minimum
deletion distance the vertex cover number of $G$.
\end{remark}

We show that isomorphism is fixed-parameter
tractable on such graphs parameterized by $k$ with fixed degree $d$;
in particular we give a procedure that computes a  polynomial kernel
for the deletion set in linear time.

\begin{theorem} \label{T:deldistance_kernel}
For any graph G and integers $d,k > 0$, we can identify in linear time a
subgraph $G'$ of $G$, a set of vertices $U \subseteq V(G')$ with
$|U| = O(k(k+d)^2)$ and a $k' \leq k$ such that: $G$ has deletion distance $k$
to degree $d$ if and only if $G'$ has deletion distance $k'$ to $d$ and,
moreover, if $G'$ has deletion distance at most $k'$, then any minimum size
$d$-deletion set for $G'$ is contained in $U$.
\end{theorem}
\begin{proof}
  Let $H := \{v \in V(G) \mid \deg(v) > k + d\}$.  Now, if $R$ is a 
  minimum size $d$-deletion set for $G$ and $G$ has deletion distance at
  most $k$ to degree $d$, then $|R| \leq k$ and the
  vertices in $V(G \setminus R)$ have degree at most $k + d$ in $G$.  So
  $H \subseteq R$.
  This means that if $|H| > k$,  then $G$ must have deletion distance greater than
  $k$ to  degree $d$ and in that case we let $G' := G, k' := k$ and $U
  = \emptyset$.

  Otherwise let $G' := G \setminus H$ and $k' := k - |H|$. We have
  shown that every $d$-deletion set of size at most $k$ must contain
  $H$. Thus $G$ has deletion distance $k$ to degree $d$ if
  and only if $G'$ has deletion distance $k'$ to degree $d$.

  Let $S := \{v \in V(G') \mid \deg_{G'}(v) > d\}$ and $U := S \cup
  N_{G'}(S)$.  Let $R' \subseteq V(G')$ be a minimum size $d$-deletion
  set for $G'$. We show that $R' \subseteq U$. Let $v \not\in U$. Then
  by the definition of $U$ we know that $\deg_{G'}(v) \leq d$ and all
  of the neighbours of $v$ have degree at most $d$ in $G'$. So if $v
  \in R'$, then $G \setminus (R' \setminus \{v\})$ also has maximal
  degree $d$, which contradicts the assumption that $R'$ is of minimum
  size. Thus $v \not\in R'$.

  Note that the vertices in $G' \setminus (R' \cup N(R'))$ have the
  same degree in $G'$ as in $G$ and thus all have degree
  at most $d$. So $S \subseteq R' \cup N(R')$ and thus $|U|
  \leq k' + k'(k+d) + k'(k+d)^2 = O(k(k+d)^2)$.

  Finally, the sets $H$ and $U$ defined as above can be found in linear time, and
  $G', k'$ can be computed from $H$ in linear time.
\end{proof}

\begin{remark}
Note that if $U = \emptyset$ and $k' > 0$, then there are no
$d$-deletion sets of  size at most $k'$.
\end{remark}

Next we see how the kernel $U$ can be used to determine whether two
graphs with deletion distance $k$ to degree $d$ are isomorphic by
reducing the problem to isomorphism of coloured graphs of degree
at most $d$.

Suppose we are given two graphs $G$ and $H$ with $d$-deletion sets $S
= \{v_1, \dots, v_k\}$ and $T = \{w_1, \dots, w_k\}$ respectively.
Further suppose that the map  $v_i \mapsto w_i$ is an isomorphism on
the induced subgraphs $G[S]$ and $H[T]$.  We can then test if this map
can be extended to an isomorphism from $G$ to $H$ using
Theorem~\ref{T:bdd_canon}.  To be precise, we define the coloured graphs
$G'$ and $H'$ which are obtained from $G\setminus S$ and $H\setminus
T$ respectively, by colouring vertices.  A vertex $u \in V(G')$ gets
the colour  $\{i \mid v_i \in N_G(u)\}$, i.e.\ the set of indices of its
neighbours in $S$.  Vertices in $H'$ are similarly coloured by the
sets of indices of their neighbours in $T$.  It is clear that $G'$ and $H'$ are
isomorphic if, and only if, there is an isomorphism between $G$ and
$H$, extending the fixed map between $S$ and $T$.  The coloured graphs
$G'$ and $H'$ have degree bounded by $d$, so Theorem~\ref{T:bdd_canon}
gives us a polynomial-time isomorphism test on these graphs.

Now, given a pair of graphs $G$ and $H$ which have deletion distance
$k$ to degree $d$, let $A$ and $B$ be the sets of vertices of degree
greater than $k+d$ in the two graphs respectively.  Also, let $U$ and
$V$ be the two kernels in the graphs obtained from
Theorem~\ref{T:deldistance_kernel}.  Thus, any $d$-deletion set in $G$
contains $A$ and is contained in $A \cup U$ and similarly, any
$d$-deletion set for $H$ contains $B$ and is contained in $B \cup V$.
Therefore to test $G$ and $H$ for isomorphism, it suffices to consider
all $k$-element subsets $S$ of $A \cup U$ containing $A$ and all
$k$-element subsets $T$ of $B\cup V$ containing $B$, and if they are
$d$-deletion sets for $G$ and $H$, check for all $k!$ maps
between them whether the map can be extended to an isomorphism from
$G$ to $H$.  As $d$ is constant this takes time
$O^*\left({{k^3}\choose{k}}^2 \cdot k!\right)$, which is
$O^*\left(2^{7k\log k}\right)$.

\section{Elimination distance to bounded degree}
\label{S:elimination_distance}

In this section we introduce a new structural parameter for graphs. We
generalise the idea of deletion distance to triviality by
recursively allowing deletions from each component of the graph. This
generalises the idea of elimination height or tree-depth, and is
equivalent to it when the notion of triviality is the empty
graph. In the context of graph isomorphism and canonisation we again
define triviality to mean bounded degree, so we look at the
elimination distance to bounded degree.

\begin{definition}
  The \emph{elimination distance to degree $d$} of a graph $G$ is
  defined as follows:
  { \small
  \[
  \textstyle{\edd(G)} :=
  \begin{cases}
    0,
    & \text{if }\Delta(G) \leq d; \\
    1 + \min \{\edd(G \setminus v) \mid v \in V(G)\},
    & \text{if $\Delta(G) > d$ and $G$ is connected;} \\
    \max\{\edd(H) \mid H \text{ a connected component of $G$}\},
    & \text{otherwise.}
  \end{cases}
  \]
  }
\end{definition}

We first introduce other
equivalent characterisations of this parameter. If $G$ is a graph
that has elimination distance $k$ to degree $d$, then we can associate
a certain tree order $\leq$ with it:

\begin{definition} \label{D:elim_order_to_deg}
A tree order $\leq$ on $V(G)$ is an \emph{elimination order to degree $d$} for
$G$ if for each $v \in V(G)$ the set
\[ 
S_v := \{u \in V(G) \mid uv \in E(G) \text{ and } u \not\leq v \text{ and } v \not\leq u\}
\] 
satisfies either:
\begin{itemize}
\item $S_v = \emptyset$; or
\item $v$ is ${\leq}$-maximal, $|S_v| \leq d$, and for all $u \in
  S_v$, we have $\{w \mid w < u\} = \{w \mid w < v\}$.
\end{itemize}
\end{definition}
\begin{remark}
Note that if $S_v = \emptyset$ for all $v \in V(G)$, then an
elimination order to degree $d$ is just an elimination order, in the
sense of Definition~\ref{def:elimination-order}.
\end{remark}

\begin{proposition}\label{prop:tree-depth-order}
A graph $G$ has $\edd(G) \leq k$ if, and only if, there is an elimination
order to degree $d$ of height $k$ for $G$.
\end{proposition}
\begin{proof}
Let $S_v$ be as in Definition~\ref{D:elim_order_to_deg}.
We prove the proposition by induction on $k$. If $k = 0$, then the graph has no vertex
of degree larger than $d$ and we define the elimination order $\leq$ to be the identity
relation on $V(G)$.
Then every $v \in V(G)$ is maximal, we have $|S_v| \leq d$, and for all
$u \in S_v$ we have $\{w \mid w < u\} = \emptyset = \{w \mid w < v\}$.

Suppose $k>0$ and the statement is true for smaller values. If $G$ is not
connected, we apply the following argument to each component. So in
the following we assume that $G$ is connected.

Suppose $\edd(G) = k$. Then there is a vertex $a \in V(G)$ such that
the components $C_1, \dots, C_r$ of $G \setminus a$ all have
$\edd(C_i) \leq k-1$. So by the induction hypothesis each $C_i$ has
a tree order $\leq_i$ to degree $d$ of height at most $k-1$ with the
properties in Definition~\ref{D:elim_order_to_deg}.
For each $v \in V(C_i)$ define 
\[ 
S_v^i := \{u \in V(C_i) \mid uv \in E(G) \text{ and } u \not\leq v \text{ and } v \not\leq u\}.  
\] 

Let
\[
{\leq} := \{(a, w) \mid w \in V(G)\} \cup \bigcup_i \leq_i. 
\]
Then $\leq$ is clearly a tree order for $G$. Note that $S_a =
\emptyset$. Let $v \in V(G) \setminus a$ be a vertex different
from $a$, say $v \in V(C_i)$. Note $S_v^i = S_v$.
If $S_v \neq \emptyset$, then $v$ is
$\leq_i$-maximal,  and thus also $\leq$-maximal. Moreover, $|S_v^i| =
|S_v| \leq d$. Lastly for any $u \in S_v$:
\[
\{w \mid w < u\} = \{a\} \cup \{w \mid w <_i u\} = \{a\} \cup \{w \mid w <_i v\} = \{w \mid w < v\}.
\]

Conversely assume there is an elimination order $\leq$ to degree $d$ of height
$k$ for $G$. There is a single
minimal element $v$ of $\leq$ because $G$ is
connected and $k > 0$.
Note that $\leq$ restricted to a component $C$ of $G \setminus
v$ has height $k-1$ and thus by the induction assumption we have that
$\edd(C) \leq k-1$.
\end{proof}

We can split a graph with an elimination order to degree $d$ in two
parts: one of low degree, and one with an elimination order defined on
it. So if $G$ is a graph that has elimination distance $k$ to degree
$d$, we can associate an elimination order $\leq$ for a subgraph $H$
of $G$ of height $k$ with $G$, so that each component of $G \setminus
V(H)$ has degree at most $d$ and is connected to $H$ along just one
branch (this is defined more formally below).

\begin{proposition} \label{P:elim_order_char2}
Let $G$ be a graph and $\leq$ an elimination order to degree
$d$ for $G$ of height $k$.   If $A$ is the set of vertices in $V(G)$
that are not $\leq$-maximal, then:
\begin{enumerate}
\item $\leq$ restricted to $A$ is an elimination order of height $k-1$
  of $G[A]$; and
\item $G \setminus A$ has degree at most $d$;
\item if $C$ is the vertex set of a component of $G \setminus A$, and
  $u, v \in A$ are $\leq$-incomparable, then either $E(u, C) = \emptyset$
  or $E(v, C) = \emptyset$.
\end{enumerate}
\end{proposition}
\begin{proof}
As any $v\in A$ is non-maximal, by
Definition~\ref{D:elim_order_to_deg}, $S_v = \emptyset$. Hence if
there is an edge between $u,v \in A$, either $u < v$ or $v< u$, and
(1) follows.

Since $G\setminus A$ contains the $\leq$-maximal elements, they are
all incomparable.  By definition of an elimination order to degree
$d$, this means that each vertex in $G\setminus A$ has at most $d$
neighbours in $G\setminus A$, so this graph has degree at most $d$,
establishing (2).  

To show (3), let $C$ be the vertex set of a component of $G \setminus
A$ and let $u, v \in A$ be such that $E(u, C) \neq \emptyset$ and
$E(v, C) \neq \emptyset$.  Then there
are $a, b \in C$ such that $au, bv \in E(G)$.  By
Definition~\ref{D:elim_order_to_deg}, $u<a$ and $v<b$.  Moreover,
there is a path from $a$ to $b$ through $C$ and as all vertices along
this path $\leq$-maximal, if $(a',b')$ is an edge in the path, it must
be that $\{w \mid w < a'\} = \{w \mid w < b'\}$.  By transitivity, 
$\{w \mid w < a\} = \{w \mid w < b\}$, and so $u < b$ and $v < a$.
Since $\leq$ is a tree-order, the set $\{w \mid w < a\}$ is linearly
orderd and we conclude that $u$ and $v$ are comparable.

\end{proof}

We also have a converse to the above in the following sense.
\begin{proposition}\label{P:char2-converse}
Suppose $G$ is a graph with $A \subseteq V(G)$ a set of vertices and $\leq_A$ an
elimination order of $G[A]$ of height $k$, such that:
\begin{enumerate}
\item $G \setminus A$ has degree at most $d$;
\item if $C$ is the vertex set of a component of $G \setminus A$, and
  $u, v \in A$ are incomparable, then either $E(u, C) = \emptyset$
  or $E(v, C) = \emptyset$.
\end{enumerate}
Then, $\leq_A$ can be extended to an elimination order to degree $d$
for $G$ of height $k+1$.
\end{proposition}
\begin{proof} 
 Let
\begin{align*}
{\leq} := &{\leq_A} \cup \{(v, v) \mid v \in (V(G) \setminus A)\} \\
&\cup
\{(u, v) \mid
  u \in A, v \in C, 
  \text{$C$ a component of $G \setminus A$}, E(w, C)\neq \emptyset 
  \text{ for some $u \leq w$}\}.
\end{align*}
Then it is easily seen that $\leq$ is a tree order on $G$.   Indeed,
$\leq_A$ is, by assumption, a tree order on $A$ and for any $v \in
V(G)\setminus A$, assumption 2 guarantees that $\{w \mid w \leq v\}$
is linearly ordered.

Let $v \in V(G)$ and let $S_v$ be as in
Definition~\ref{D:elim_order_to_deg}.
Suppose $S_v \neq \emptyset$. Then $v \in (V(G) \setminus A)$ and has
degree at most $d$ in $G \setminus A$. By the construction $v$ is
$\leq$-maximal. Let $u \in S_v$. Then there is a component $C$ of $G
\setminus A$ that contains both $u$ and $v$ and thus $\{w \mid w < u\}
= \{w \mid w < v\}$. 
\end{proof}
\begin{remark}
In the following, given a graph $G$ and an elimination order to degree
$d$, $\leq$, we call the subgraph of $V(G)$ induced by the non-maximal elements
of the order $\leq$ the \emph{non-maximal subgraph of $G$ under $\leq$}.

In the proof of Proposition~\ref{P:char2-converse} above,
a suitable tree order on a subset $A$ of $V(G)$ is extended to an elimination
order to degree $d$ of $G$ by making all vertices not in $A$ maximal
in the order.  This is a form of construction we use repeatedly below.
\end{remark}

The alternative characterisations of elimination order to degree $d$
established above are very useful.  
In the next section, we use them to construct a
\emph{canonical} elimination order to degree $d$ of $G$, based on an
elimination order of a graph we call the \emph{torso} of $G$, which
consists of the high-degree vertices of $G$, along with some additional
edges.

\section{Canonical Elimination Order to Bounded Degree}
The aim of this section is to show that if a graph $G$ has elimination
distance $k$ to degree $d$, then there is an elimination order to
degree $d$ whose height is still bounded by a function of $d$ and $k$
and which is \emph{canonical}.  To be precise, we identify a graph
which we call the $d$-degree torso of $G$, which contains all the
vertices of $G$ of degree more than $k$ and has additional edges to
represent paths between these vertices that go through the rest of
$G$.  We show that this torso necessarily has tree-depth bounded by a
function of $k$ and $d$ and the canonical elimination order witnessing
this can be extended to an elimination order to degree $d$ of $G$.
The result is established through a series of lemmas.  A pattern of
construction that is repeatedly used here is that we define a certain
set $A$ of vertices of $G$ and construct an elimination order of
$G[A]$.  It is then shown that extending the order by making all
vertices in $V(G)\setminus A$ maximal yields an elimination order to
degree $d$ of $G$.  Necessarily, in this extended order, all the
non-maximal elements are in $A$.

The following lemma establishes that if $G$ has elimination distance
$k$ to degree $d$ and moreover the degree of $G$ is at most $k+d$,
then we can construct an alternative elimination order on $G$ in which
all the vertices of degree greater than $d$ are included in the
non-maximal subgraph and the height of the new elimination order is
still bounded by a function of $k$ and $d$.
\begin{lemma} \label{L:adding_vertices}
Let $G$ be a graph with maximal degree $\Delta(G) \leq k+d$. Let
$\leq$ be an elimination order to degree $d$ of height $k$ of $G$ with
non-maximal subgraph $H$, and 
let $A = V(H) \cup \{v \in V(G) \mid \deg_G(v) > d\}$.  Then $G$ has an
elimination order $\sqsubseteq$ of height at most $k(k+d+1)$ for which
the non-maximal elements are in $A$.
\end{lemma}
\begin{proof}
Let $G, H, A$ and $\leq$ be as in the statement of the lemma.
We will adapt $\leq$ to an elimination order $\sqsubseteq$ of $G[A]$.

Let $B$ be the set of $\leq$-maximal elements in $V(H)$.
For each $w \in A \setminus V(H)$ let $C_w$ be the component of $G
\setminus V(H)$ that contains $w$. Note that $N(C_w) \neq \emptyset$,
because $\deg(w) > d$, so at least one vertex in $H$ must be adjacent
to $w$. By Definition~\ref{D:elim_order_to_deg}, all
vertices in $N(C_w)$ are $\leq$-comparable, so they are linearly ordered and
there is a unique $b \in B$ such that $b \geq a$ for all $a \in
N(C_w)$.  We write $b(w)$ to denote this element of $B$ associated
with every $w \in A \setminus V(H)$.
For each $b \in B$, let $W_b := \{w \in A \setminus V(H) \mid b(w) =
b\}$, and let $\sqsubseteq_b$ be an arbitrary linear order on
$W_b$.

For any $u, v \in V(G)$, define $u \sqsubseteq v$ if one of the
following holds:
\begin{itemize}
\item $u = v$;
\item $u, v \in H$ and $u \leq v$;
\item $u \in H$, $v \in G \setminus A$ and $u \leq v$;
\item $u \in H$, $v \in A\setminus V(H)$ and $u \leq b(v)$;
\item $u \in A \setminus V(H)$, $v \in G \setminus A$ and $b(u)
  \leq v$;
\item $u, v \in H'\setminus V(H)$, $b(u)= b(v)$ and $u \sqsubseteq_b v$.
\end{itemize}

It follows from the construction that $\sqsubseteq$ restricted to $A$
is an elimination order of $G[A]$, and that $\sqsubseteq$ is an
elimination order to degree $d$ of $G$.

For each $b \in B$, the set $\{v \in H \mid v \leq b\}$ has at most
$k$ elements, by the assumption on the height of the order $\leq$.
Since $G$ has maximum degree $k+d$ and  $W_b \subseteq
N(\{v \in H \mid v \leq b\})$, we have that $W_b$ has at most $k(k+d)$
vertices.  Since the height of any $\sqsubseteq$ chain is at most the
height of a $\leq$-chain plus $|W_b|$, we conclude that 
the height of $\sqsubseteq$ is at most $k(k+d+1)$.
\end{proof}

The lemma above allows us to re-arrange the elimination order so that
it includes all vertices of large degree.  In contrast, the next lemma
gives us a means to re-arrange the elimination order so that all
vertices of small degree are made maximal in the order.  This is again
done achieved while keeping the height of the elimination order
bounded by a function of $k$ and $d$.

\begin{lemma} \label{L:removing_vertices}
Let $G$ be a graph. Let $\leq$ be an elimination order to degree
$d$ of $G$ of height $k$ with non-maximal subgraph $H$, such that $H$
contains all vertices of degree greater than $d$, and 
let $A = \{v \in V(H) \mid \deg_G(v) > d\}$. Then, there is an
elimination order to degree $d$ of $G$ of height at most 
$k((k+1)d)^{2^k}+1$  for which all the non-maximal elements are in $A$.
\end{lemma}
\begin{proof}
  Let $G, H, A$ and $\leq$ be as in the statement of the lemma.
  We assume that $G$ is connected -- if not, we can apply the argument
  to each component of $G$.
  We construct an elimination order $\sqsubseteq$ of $G[A]$ from
  $\leq$, making sure that it has height at most $k((k+1)d)^{2^k}$.
  This extends to an elimination order to degree $d$ of $G$ by making
  all vertices not in $A$ maximal, as in the
  Proposition~\ref{P:char2-converse}.

  Let $J := H \setminus A$.
  
  For $v \in V(J)$, let $K_v$ be the set of vertices $w \in A$
  such that:
  \begin{enumerate}
  \item $v \leq w$; 
  \item there is a path from $v$ to $w$ through $G \setminus A$; and
  \item for any $u$ with $u < v$, there is no path from $u$ to $w$
    through $G \setminus A$.
  \end{enumerate}
  Note that because $\leq$ is a tree order and the third condition,
  the sets $K_v$ are pairwise disjoint.
  Let $\overline K := A \setminus (\bigcup_{v \in V(J)}
  K_v)$ be the set of vertices in $A$ that are not contained in
  $K_v$ for any $v$.

  For each $v \in V(J)$, let $\sqsubseteq_v$ be an arbitrary linear
  order on $K_v$.    The idea behind the construction below is that we
  replace  $v$ in the elimination order by $K_v$, ordered by $\sqsubseteq_v$.
Formally, for any $u, w \in V(G)$, define $u \sqsubseteq w$ if
  one of the following holds:
  \begin{itemize}
  \item $u = w$;
  \item $u \in K_v$, $w \in G \setminus A$ and $v \leq w$;
  \item $u \in \overline K$, $w \in G \setminus A$ and $u \leq w$;
  \item $u, w \in K_v$ and $u \sqsubseteq_v w$;
  \item $u \in K_v$, $w \in K_{v'}$ and $v < v'$;
  \item $u \in \overline K$, $w \in K_v$ and $u \leq v$;
  \item $u \in K_v$, $w \in \overline K$ and $v \leq w$;
  \item $u, w \in \overline K$ and $u \leq w$.
  \end{itemize}

  We first show that $\sqsubseteq$ is an elimination order for $G[A]$. The
  construction ensures $\sqsubseteq$ is a tree order. 
  Let $u, w \in V(H')$. We show that if $u \leq w$, then either $u
  \sqsubseteq w$ or $w \sqsubseteq u$.
  We go through all possible
  cases: If $u = w$, we have $u \sqsubseteq w$. If there is some
  $v \in V(J)$ such that $u, w \in K_v$, then $u \sqsubseteq w$ or $w
  \sqsubseteq u$. If $u \in K_v$, $w \in K_{v'}$ for two different $v,
  v' \in V(J)$, then $v \leq u \leq w$ and $v' \leq w$, so $v' \leq v$
  and thus $w \sqsubseteq u$. If $u \in \overline K$ and $w \in K_v$,
  then both $u, v \leq w$, so either $u \leq v$ or $v \leq u$, and
  thus either $u \sqsubseteq w$ or $w \sqsubseteq u$. The case where
  $u \in K_v$, $w \in \overline K$ is symmetric. Finally, if both $u,
  w \in \overline K$, then $u \sqsubseteq w$. Thus if $uw \in E(H')$,
  we have $u \leq w$ or $w \leq u$ and therefore $u \sqsubseteq w$ or
  $w \sqsubseteq u$. Hence $\sqsubseteq$ is an elimination order for
  $G[A]$.

  Let $Z$ be a component of $G \setminus A$. 
  We assumed that $H$ contains all vertices of degree greater than
  $d$, and by the construction $A$ also contains all those
  vertices. Thus $Z$ has maximum degree $d$.

  Suppose $u, v \in A$ are two vertices that are connected to $Z$,
  i.e. $E_G(u, V(Z)) \neq \emptyset \neq E_G(v, V(Z))$. We
  show that either $u \sqsubseteq v$ or $v \sqsubseteq u$.
  Note that there is a path $P$ through $Z \subseteq G \setminus A$
  from $u$ to $v$, i.e.\ all vertices in $P$, except for the endpoints,
  lie outside of $A$.
  If $P$ contains no vertices from $J$, then
  the connected component $Z'$ of $G \setminus V(H)$ containing $P \setminus
  \{u,v\}$ satisfies $E_G(u, V(Z')) \neq \emptyset \neq E_G(v, V(Z'))$
  and thus $u \leq v$ or $v \leq u$, and therefore by the above $u
  \sqsubseteq v$ or $v \sqsubseteq u$.
  
  Otherwise, $P$ contains vertices from $J$. Let $w$
  be a $\leq$-minimal vertex in $V(P) \cap V(J)$. Then there
  is a path outside of $A$ from $w$ to $u$, and also to $v$ (both part of
  $P$). Moreover, if neither $u \leq v$ nor $v \leq u$, then $w \leq u$
  and $w \leq v$.
  Thus $u$ and $v$ are in $K_w$ (or in $K_{w'}$ for some $w' < w$), and
  therefore $u \sqsubseteq v$ or $v \sqsubseteq u$.

  It remains to show that the size of $K_v$ is bounded by $k((k+1)d)^{2^k}$
  for all $v \in V(J)$. Let $G'$ be the graph obtained from $G$
  by adding an edge between two vertices $s,t \in V(J)$ whenever there is a
  path through $G \setminus V(H)$ between $s$ and $t$. This increases the
  degree of vertices in $V(J)$ by at most $kd$, because each of these vertices is
  connected to at most $d$ components of $G \setminus V(H)$ and each of
  these is connected to at most $k$ vertices in $H$. Now there is a
  path between two vertices $s,t \in J$ in $G$ outside of $A$ if and only if there
  is a path between $s$ and $t$ in $G'[V(J)]$. Moreover, $\leq$ is
  also an elimination order for $G'[V(J)]$.
  So, as $G'[V(J)]$ has tree-depth at most $k$ it does not contain
  a path of length more than $2^k$. Since each vertex on the path
  has degree at most $(k+1)d$, we can reach at most $((k+1)d)^{2^k}$
  vertices in $A$ on paths only containing vertices outside of
  $A$. Thus $|K_v| \leq ((k+1)d)^{2^k}$ and the height of
  $\sqsubseteq$ is bounded by $k|K_v| \leq k((k+1)d)^{2^k}$.
\end{proof}

Next we introduce the notion of $d$-degree torso and prove that it
captures the properties that we require of an elimination tree to
degree $d$.

\begin{definition}
Let $G$ be a graph, let $d > 0$ and let $H$ be the induced subgraph of
$G$ containing the vertices of degree larger than $d$.
The \emph{$d$-degree torso of $G$} is the graph $C$ obtained from $H$
by adding an edge between two vertices $u,v \in H$ if there is a path
through $G \setminus V(H)$ from $u$ to $v$ in $G$.

\end{definition}

The next lemma establishes an upper bound on the tree-depth of the
torso of a graph when the maximum degree is bounded.

\begin{lemma} \label{L:torso_and_tree-depth}
Let $G$ be a graph and let $C$ be the $d$-degree torso of $G$. Let $H
= G[V(C)]$ and let ${\leq}$ be an elimination order for $H$. Then ${\leq}$
is an elimination order for $C$ of height $h$ if, and only if, ${\leq}$ can be extended
to an elimination order to degree $d$ for $G$ of height $h+1$.
\end{lemma}
\begin{proof} 
Let $G, C, H$ and $\leq$ be as above.

Suppose $\leq$ is an elimination order for $C$. Since $C$ is a
supergraph of $H$, this means that $\leq$ is an elimination order for
$H$. Let $Z$ be a component of $G \setminus V(H)$. Since $C$ contains
all vertices of degree greater than $d$, $Z$ has maximal degree $d$.
If $E(Z, u) \neq \emptyset$ and $E(Z, v) \neq \emptyset$ for two
vertices $u, v \in H$, then there is a path through $Z \subseteq G
\setminus V(H)$ connecting $u$ and $v$, so by the definition of the
$d$-degree torso $uv \in E(C)$ and thus $u, v$ are $\leq$-comparable. We can
extend $\leq$ to a tree order $\leq'$ on $V(G)$ where all the vertices from
$V(G) \setminus V(H)$ are maximal.

Conversely assume that $\leq$ can be extended to an elimination order
to degree $d$ for $G$.
Let $uv \in E(C)$. If $uv \in E(H)$, then $u$ and $v$ must be
$\leq$-comparable. Otherwise $uv \not\in E(H)$, so there is a path
through $G \setminus V(H)$ from $u$ to $v$ in $G$, i.e. both $u$ and $v$ are
connected to a component $Z$ of $G \setminus V(H)$ and thus
comparable. Therefore $\leq$ is an elimination order for $C$.
\end{proof}

\begin{lemma} \label{L:bounded_case}
Let $G$ be a graph with elimination distance $k$ to degree $d$ and
maximum degree $\Delta(G) \leq k + d$. Let $C$
be the $d$-degree torso of $G$ and let $\leq$ be a minimum height
elimination order for $C$. Then $\leq$ has height at most
$k(k+d+1)((k(k+d+1)+1)d)^{2^{k(k+d+1)}}$.

\end{lemma}
\begin{proof} 
Let $\sqsubseteq$ be a minimum height elimination order to degree $d$ of
$G$. Since $G$ has elimination distance to degree $d$ at most $k$, the
height of $\sqsubseteq$ is at most $k$. Let $H$ be the non-maximal
subgraph of $G$ under  $\sqsubseteq$ and define
\[
A = V(H) \cup \{v \in V(G) \mid \deg_G(v) > d\}.
\]
By Lemma~\ref{L:adding_vertices}, the graph $G[A]$ has an elimination order
$\preceq$ of height at most $k(k+d+1)$ that can be extended to an
elimination order to degree $d$ for $G$.

Let $A' = \{v \in V(H') \mid \deg_G(v) > d\}$. By
Lemma~\ref{L:removing_vertices}, the graph $A'$ has an elimination order
$\leq$ of height at most $k(k+d+1)((k(k+d+1)+1)d)^{2^{k(k+d+1)}}$ that
can be extended to an elimination order to degree $d$ for $G$.

Lastly note that $A' = V(C)$, so that by
Lemma~\ref{L:torso_and_tree-depth}, $\leq$ is an elimination order
for $C$.
\end{proof}

We are now ready to prove the main result:

\begin{theorem}
  Let $G$ be a graph that has elimination distance $k$ to degree $d$.
  Let $\leq$ be a minimum height elimination order of the $d$-degree torso
  $G$. Then $\leq$ can be extended to an elimination order to degree
  $d$ of $G$ of height at most
    \[
 k((k+1)(k+d))^{2^{k}} + k(1+k+d)(k(1+k+2d))^{2^{k(1+k+d)}}+1.
    \]
\end{theorem}
\begin{proof}
  We show that the $d$-degree torso of $G$ has an elimination order of
  height at most $ k((k+1)(k+d))^{2^{k}} + k(1+k+d)(k(1+k+2d))^{2^{k(1+k+d)}}$.  The Theorem then follows by Lemma~\ref{L:torso_and_tree-depth}.

  Let $C$ be the $(k+d)$-degree torso of $G$. We first show that the
  tree-depth of $C$ is bounded by $k((k+1)(k+d))^{2^k}$. To see this, let
  $\sqsubseteq$ be an elimination order to degree $d$ of $G$ of minimum
  height with non-maximal subgraph $H$.
  Note that $H$ contains all vertices of degree greater than $k+d$,
  because vertices in $G \setminus V(H)$ are adjacent to at most $k$
  vertices in $H$.

  Let $A = \{v \in V(H) \mid \deg_G(v) > k+d\}$.
  By Lemma~\ref{L:removing_vertices}, the graph $G[A]$ has an elimination order
  $\preceq$ of depth at most $h:= k((k+1)(k+d))^{2^k}$ that can be extended to an
  elimination order to degree $k+d$ of $G$ of height $h+1$.
  Note that $A = V(C)$, so by Lemma~\ref{L:torso_and_tree-depth},
  the order $\preceq$ is an elimination order for $C$.  Let $\preceq'$
  denote its extension to $G$.

  Let $Z$ be a component of $G \setminus A$ and let $C_Z$ be the
  $d$-degree torso of $Z$. By Lemma~\ref{L:bounded_case}, there is an
  elimination order $\preceq_Z$ for $C_Z$ of height at most
  $k(k+d+1)((k(k+d+1)+1)d)^{2^{k(k+d+1)}}$. Let $v_Z$ be the $\preceq$-maximal
  element in $C$ such that there is a $w \in C_Z$ with $v_Z \preceq'
  w$. Define
  \begin{align*}
    \leq' := &\preceq \cup \bigcup_Z \preceq_Z
             \cup \bigcup_Z \{(v,w) \mid v \preceq' v_Z, w \in C_Z \}.
  \end{align*}
Observe that $C \cup \bigcup_{Z} C_Z$ is a subgraph of the $d$-degree torso of 
  $G$. Thus $\leq'$ is an elimination order for the $d$-degree torso of 
  $G$.  The height of $\leq'$ is bounded by
  \begin{align*}
  \td(C) + \max \{\td(C_Z)\}_Z
  &\leq k((k+1)(k+d))^{2^{k}} + k(1+k+d)(k(1+k+2d))^{2^{k(1+k+d)}}.
  \end{align*}
\end{proof}

\section{Canonisation parameterized by elimination distance to bounded degree}
\label{S:elimination_distance_alg}

In this section we show that graph canonisation, and thus graph
isomorphism, is $\FPT$ parameterized
by elimination distance to bounded degree.  The main idea is to
construct a labelled directed tree $T_G$ from a graph $G$ (of
elimination distance $k$ to degree $d$) that is an isomorphism
invariant for $G$.  From the labelled tree $T_G$ we obtain a canonical
labelled tree using the tree canonisation algorithm from
Lindell~\cite{lindell_logspace_1992}. In the last step we construct a canonical
form of $G$ from the canonical labelled tree.

The tree $T_G$ is obtained from $G$ by taking a tree-depth
decomposition of the $d$-degree torso of $G$ and labelling the nodes
with the isomorphism types of the low-degree components that attach to
them.  The tree-depth decomposition of a graph is just the elimination
order in tree form.
We formally define it as follows:
\begin{definition}\label{def:tree-depth-decomp}
Given a graph $H$ and an elimination order $\leq$ on $H$, the
\emph{tree-depth decomposition} associated with $\leq$ is the directed
tree with nodes $V(H)$ and an arc $a\rightarrow b$ if, and only if, $a
< b$ and there is no $c$ such that $a < c < b$.
\end{definition}
\begin{remark}
The tree-depth decomposition corresponding to an elimination order is
what, in the language of partial orders, is known as its covering relation.
\end{remark}

Note that, in general, the tree-depth decomposition of a graph that is
not connected may be a forest.  By results of
Bouland~\emph{et. al}~\cite{bouland_tractable_2012}, we can construct
a canonical tree-depth decomposition of an $n$-vertex graph of tree-depth $k$ in
time $f(k) \cdot n^c$ for some comuptable $f$ and constant $c$.

Before defining $T_G$ formally, we need one piece of terminology.
\begin{definition}
Let $G$ be a graph and let $\leq$ be a tree order for $G$. The
\emph{level} of a vertex $v \in V(G)$ is the length of the
chain $\{w \in V(G) \mid w \leq v\}$. We denote the
level of $v$ by $\level_\leq(v)$.
\end{definition}

Given a graph $G$ of elimination distance $k$ to degree $d$, let $C$
be the $d$-degree torso of $G$, let $T$ be a canonical tree-depth
decomposition of $C$ and $\leq$ the corresponding elimination order.
Let $Z$ be a component of $G\setminus C$.  We let $Z^C$ denote the
coloured graph that is obtained by colouring each vertex $v$ in $Z$ by
the colour $\{ i \mid uv \in E(G) \text{ for some } u \in C \text{
  with } \level_\leq(u) = i\}$.  We write $F(Z^C)$ for the canonical form
of this coloured graph given by Theorem~\ref{T:bdd_canon}.  Note that,
by the definition of elimination distance, there is, for each $Z$ and
$i$ at most one vertex $u \in C$ with $\level_\leq(u) = i$ which is in
$N_G(Z)$. 

We are now ready to define the labelled tree $T_G$.
The nodes of $T_G$ are the nodes of $T$ together with a new node $r$,
and the arcs are the arcs of $T$ along with new arcs from $r$ to the
root of each tree in $T$.  Define, for each node $u$ of $T_G$,
$\mathcal{Z}_u$ to be the set $\{Z \mid Z \text{ is a component
  of } G\setminus C \text{ with } u \leq\text{-maximal in } C \cap
N_G(Z) \}$ (if $u \neq r)$ and $\{Z \mid Z \text{ is a component of }
G\setminus C \text{ with } C \cap N_G(Z)= \emptyset \}$ (if $u=r$).
Each node $u$ in $T$ carries a label consisting of two parts:

\begin{itemize}
\item $L_w := \{level(w) \mid w < u \text{ and } uw \in E(G)\}$; and
\item the multiset $\{F(Z^C) \mid Z \in \mathcal{Z}_u  \}$.
\end{itemize}

\begin{proposition}\label{prop:canonical_tree}
For any graphs $G$ and $G'$, $T_G$ and $T_{G'}$ are isomorphic
labelled trees if, and only if, $G \cong G'$.
\end{proposition}
\begin{proof}
If $G \cong G'$ then, by construction, their $d$-degree torsos induce
isomorphic graphs.  The canonical tree-depth decomposition of Bouland
et al.\ then produces isomorphic directed trees and the isomorphism
must preserve the labels that encode the rest of the graphs $G$ and
$G'$ respectively.

For the converse direction, suppose we have an isomorphism $\phi$
between the labelled trees $T_G$ and $T_{G'}$.  Since the label $L_u$ of
any node $u$ encodes all ancestors of $u$ which are
neighbours, $\phi$ must preserve all edges and non-edges in the
$d$-degree torso $C$ of $G$.  To extend $\phi$ to all of $G$, for each
node $u$ in $T_G$, let $\beta_u$ be a bijection from
$\mathcal{Z}_u$ to the corresponding set $\mathcal{Z}_{\phi(u)}$ of
components of $G'\setminus C'$, such that $F(Z^C) =
F(\beta_u(Z)^{C'})$ (such a bijection exists as $u$ and $\phi(u)$
carry the same label).  Thus, in particular, there is an isomorphism
between $Z^C$ and $\beta_u(Z)^{C'}$, since they have the same
canonical form.  We define, for each $v \in
V(G)\setminus C$, $\phi(v)$ to be the image of $v$ under the
isomorphism taking the component $Z$ containing $v$ to
$\beta_u(Z)$.  Note that this gives a well-defined function on $V(G)$,
because for each such $v$, there is exactly one node $u$ of $T_G$ such
that the component containing $v$ is in $\mathcal{Z}_u$.  We claim
that $\phi$ is now an isomorphism from $G$ to $G'$.  Let $vw$ be an
edge of $G$.  If both $v$ and $w$ are in $C$, then either $v< w$ or $w
< v$.  Assume, without loss of generality, that it is the former.
Then, $\level(v) \in L_w$ is in the label of $w$ in $T_G$ and since
$\phi$ is a label-preserving isomorphism from $T_G$ to $T_{G'}$,
$\phi(v)\phi(w)$ is an edge in $G'$.  If both $v$ and $w$ are in $G
\setminus C$, then there is some component $Z$ of $G \setminus C$ that
contains them both.  Since $\phi$ maps $Z$ to an isomorphic component
of $G' \setminus C'$, $\phi(v)\phi(w) \in E(G')$.  Finally, suppose
$v$ is in $C$ and $w$ in $G\setminus C$ and let $Z$ be the component
containing $w$.  Then $i := \level(v)$ is part of the colour of $w$ in
$Z^C$ and hence part of the colour of $\phi(w)$ in the corresponding
component of $G'\setminus C'$.  Moreover, if $u$ is the $\leq$-maximal
element in $C \cap N_G(Z)$, then we must have $v \leq u$.  Thus
$\phi(v)$ is the unique element of level $i$ in $C' \cap
N_{G'}(\beta_u(Z))$ and we conclude that $\phi(v)\phi(w) \in E(G')$.
By a symmetric argument, we have that for any edge $vw \in E(G')$,
$\phi^{-1}(v)\phi^{-1}(w) \in E(G)$ and we conclude that $\phi$ is an
isomorphism. 
\end{proof}

With this, we are able to establish our main result.
\begin{theorem}
Graph Canonisation is $\FPT$ parameterized by elimination distance to 
bounded degree. 
\end{theorem}
\begin{proof}
Suppose we are given a graph $G$ with $|V(G)| = n$.
We first compute the $d$-degree torso $C$ of $G$ in $O(n^4)$ time.
Using the result from
Bouland~\emph{et. al}~\cite[Theorem 11]{bouland_tractable_2012}, we
can find a canonical tree-depth decomposition for $C$ in time
$O(h(k)n^3log(n))$ for some computable function $h$.  To compute the
labels of the nodes in the trees (and hence obtain) $T_G$, we
determine, for each $u \in C$, the set $\{level(w) \mid w
< u \text{ and } uw \in E(G)\}$.  This can be done in time $O(n^2)$.
Then, we find the components of $G\setminus C$, and colour the
vertices with the levels of their neighbours in $C$.  This can be done
in $O(n^2)$ time. Finally, we compute for each coloured component
$Z^C$ the canonical representative $F(Z^C)$ which, by
Theorem~\ref{T:bdd_canon} can be done in polynomial time (where the
degree of the polynomial depends on $d$).

Having obtained $T_G$, we compute the canonical form $T_G'$ in linear
time using Lindell's canonisation algorithm
\cite{lindell_logspace_1992}.  Using the labels of $T_G'$ one can, in
linear time, construct a graph $G'$ such that $T(G')=T_G'$.  By
Proposition~\ref{prop:canonical_tree}, this is a 
canonical form $G'$ of $G$.
\end{proof}

\begin{corollary}
Graph Isomorphism is $\FPT$ parameterized by elimination distance to 
bounded degree. 
\end{corollary}

\section{Conclusion}
\label{S:conclusion}

We introduce a new way of parameterizing graphs by their distance to
triviality, i.e.\ by elimination distance.  In the particular case of
graph canonisation, and thus also graph isomorphism, taking triviality
to mean graphs of bounded degree, we show that the problem is $\FPT$.  

A natural question that arises is what happens when we take other
classes of graphs for which graph isomorphism is known to be tractable
as our ``trivial'' classes.
For instance, what can we say about $\GI$ when
parameterized by elimination distance to planar graphs?  Unfortunately
techniques such as those deployed in the present paper are
unlikely to work in this case.  Our techniques rely on identifying a
canonical subgraph which defines an elimination tree into the trivial
class.  In the case of planar graphs, consider graphs which are
subdivisions of $K_5$, each of which is deletion distance 1 away from
planarity.  However the deletion of \emph{any} vertex yields a planar
graph and it is therefore not possible to identify a canonical such
vertex.

More generally, the notion of elimination distance to triviality seems
to offer promise for defining tractable parameterizations for many
graph problems other than isomorphism.  This is a direction that bears
further investigation.

It is easy to see that if a class of graphs $\C$ is characterised by a
finite set of excluded minors, that the class $\hat\C$ of graphs with bounded
elimination distance to $\C$ is characterised by a finite set of
excluded minors as well. An interesting question is whether we can,
given the set of excluded minors for $\C$, compute the excluded minors
for $\hat\C$ as well?

\bibliographystyle{amsplain}
\bibliography{isomorphism_elim_distance.bib}

\end{document}